\documentclass[aps,pra,epsfigure,onecolumn,superscriptaddress]{revtex4}
\pdfoutput=1
\usepackage[dvips]{graphicx}
\usepackage{amsfonts}
\usepackage{amsthm}
\usepackage{amsmath}
\usepackage{amssymb}
\usepackage{amsbsy}
\usepackage{ulem}
\usepackage[colorlinks, citecolor=red]{hyperref}
\usepackage{color}
\usepackage{array}
\usepackage{appendix}

\newcommand{\ket}[1]{\vert #1 \rangle}
\newcommand{\bra}[1]{\langle #1 \vert}
\newcommand{\operator}[1]{\mathrm{#1}}

\newcommand{\PreserveBackslash}[1]{\let\temp=\\#1\let\\=\temp}
\newcolumntype{C}[1]{>{\PreserveBackslash\centering}p{#1}}

\theoremstyle{plain}
\newtheorem{theorem}{THEOREM}

\newtheorem{proposition}[theorem]{PROPOSITION}
\newtheorem{corollary}[theorem]{COROLLARY}

\begin{document}

\title{Grover Walks on a Line with Absorbing Boundaries}

\author{Kun Wang}
\affiliation{State Key Laboratory for Novel Software Technology, Jiangsu 210093, China}
\affiliation{Department of Computer Science and Technology,
                Nanjing University, Jiangsu 210093, China}
\author{Nan Wu}
\email{nwu@nju.edu.cn}
\affiliation{State Key Laboratory for Novel Software Technology, Jiangsu 210093, China}
\affiliation{Department of Computer Science and Technology,
                Nanjing University, Jiangsu 210093, China}
\author{Parker Kuklinski}
\affiliation{Department of Mathematics, Boston University, Boston, MA 02215, USA}
\author{Ping Xu}
\affiliation{National Lab of Solid State Microstructures, Jiangsu 210093, China}
\author{Haixing Hu}
\affiliation{National Lab of Solid State Microstructures, Jiangsu 210093, China}
\author{Fangmin Song}
\affiliation{State Key Laboratory for Novel Software Technology, Jiangsu 210093, China}
\affiliation{Department of Computer Science and Technology,
                Nanjing University, Jiangsu 210093, China}
\begin{abstract}
In this paper, we study Grover walks on a line with one and two absorbing boundaries.
In particular, we present some results for the absorbing probabilities
both in a semi-finite and finite line.
Analytical expressions for these absorbing probabilities are presented by using the
combinatorial approach. These results are perfectly matched with numerical simulations.
We show that the behavior of Grover walks on a line with absorbing boundaries
is strikingly different from that of classical walks and that of Hadamard walks.
\end{abstract}

\maketitle

\section{Introduction}\label{sec:introduction}

Since the seminal work by~\cite{aharonov1993quantum},
quantum walks have been the subject of research for the past two decades.
They were originally proposed as a quantum generalization of classical
random walks~\cite{spitzer2013principles}.
Asymptotic properties such as mixing time, mixing rate and hitting time
of quantum walks on a line and on general graphs have been studied extensively
~\cite{ambainis2001one, aharonov2001quantum,
moore2002quantum, childs2004spatial, krovi2006hitting}.
Applications of quantum walks for quantum information processing have also been investigated.
Especially, quantum walks can solve the element distinctness problem
~\cite{aaronson2004quantum, ambainis2007quantum} and
perform the quantum search algorithms~\cite{szegedy2004quantum}.
In some applications, quantum walks based algorithms can even gain
exponential speedup over all possible classical algorithms~\cite{childs2003exponential}.
The discovery of their capability for universal quantum
computations~\cite{childs2009universal, lovett2010universal}
indicates that understanding quantum walks is necessary for
better understanding quantum computing itself.
For a more comprehensive review, we refer the readers to
~\cite{kempe2003quantum, venegas2012quantum} and the references within.

One dimensional three state quantum walks, first considered by
Inui et al.~\cite{inui2005one}, are variations of two state quantum walks on a line.
In the three state walk, the walker is governed by a coin with three degrees of freedom.
In each step, the walker is not only capable of moving left or right,
but also able to stay at the same position.
Three state quantum walks have interesting differences from two state quantum walks.
Most notably, if the walker of a three state quantum walk is initialized at one site,
it is trapped with large probability near the origin after walking enough steps~\cite{inui2005localization}.
This phenomenon is previously found in quantum walks on
square lattices~\cite{inui2004localization} and is called localization.
In fact, this model is the simplest model that exhibits localization,
a quantum effect entirely absent from the corresponding classical random walk.
Three state random walks are essentially regarded as
the same process as two state random walks by scaling the time.
A thorough understanding of the localization effect on this model
becomes particularly relevant given the fact that this phenomenon
is commonplace in higher-dimensional systems~\cite{venegas2012quantum}.
Recent researches showed that the localization effect happens with a broad family of coin
operators in three state quantum walks~\cite{vstefavnak2012continuous, vstefavnak2014stability, vstefavnak2014limit}.

The presence of absorbing boundaries apparently complicates the analysis of
three state quantum walks considerably.
In this paper, we focus on the question of determining the absorbing
probabilities in Grover walks with one and two boundaries.

First, we consider the case where we have a single absorbing boundary.
The walk process is terminated if the walker reaches that boundary.
We offer methods to calculate the absorbing probability for an arbitrary boundary.
When the boundary is fixed at $-1$, it is known that in the classical case
the walker is absorbed with probability 1~\cite{motwani2010randomized},
while in Hadamard walks the walker has an absorption
probability of $2/\pi$~\cite{ambainis2001one}.
Intuitively, as some probability amplitudes are trapped near origin
due to the localization effect in Grover walks, the absorbing probability
of Grover walks is smaller than that of Hadamard walks.
However, the Grover walker is absorbed with probability $0.6693$ which is larger than $2/\pi$.
What's more, when the boundary is moved from $-1$ to $-2$,
the absorbing probabilities suffer an extreme fast decrease.
To explain these strange behaviors,
we numerically study the oscillating localization effect in Grover walks with one boundary.
We find that the localization is occurred
owing to the quantum state oscillating between $-1$ and $0$.
If the boundary is at $-1$, the localization effect disappears and the state is absorbed,
resulting in a large absorbing probability.
If the boundary is at $-2$, the localization effect revives and
the absorbing probability plummets.

Then, we review the case where there are two boundaries -
one is at site $-M$ ($M > 0$) to the walker's left,
and the other is at site $N$ ($N > 0$) to the walker's right.
The walk is terminated if the walker is trapped in either absorbing boundary.
Methods are designed to calculate the left and right absorbing probabilities to
arbitrary accuracy for arbitrary left and right boundaries.
In Hadamard walks, the left absorbing probability
approaches $1/\sqrt{2}$~\cite{ambainis2001one} when the left boundary is at $-1$ and
the right boundary approaches infinity.
It is concluded that adding a second boundary on the right
actually increases the probability of reaching the left~\cite{ambainis2001one}.
In Grover walks with two boundaries,
we get the same left absorbing probability $1/\sqrt{2}$ under the same setting.
The conclusion is still correct in Grover walks as $1/\sqrt{2} > 0.6693$.
When the left boundary is at $-1$, the oscillating localization effect disappears.
As position $-1$ is occupied, the part of quantum state,
which would have otherwise localized, now is absorbed.
When the left boundary is to the left of $-1$,
the sum of the left and right absorbing probabilities is generally less than $1$
due to oscillating localization.
When studying the case where the left boundary is at $-2$,
we show that the localization probabilities are exponentially decaying
in Grover walks with two boundaries.

The rest of this paper is organized as follows.
Section \ref{sec:definition} gives formal definitions of Grover walks with
one and two absorbing boundaries, and the absorbing probabilities which we study.
Section \ref{sec:one-boundary} and Section \ref{sec:two-boundaries}
present the methods and results on Grover walks with one and two boundaries.
Finally, we conclude in Section \ref{sec:conclusion}.

\section{Definitions}\label{sec:definition}

\subsection{Grover walks}
The three state quantum walk (3QW) considered here is a kind of generalized two state quantum
walk on a line. The Hilbert space of the system is given by the tensor product
$ \mathcal{H} = \mathcal{H}_P \otimes \mathcal{H}_C$
of the position space
$$ \mathcal{H}_P = \textrm{Span}\{\ket{m}, m \in \mathbb{Z}\}$$
and the coin space $\mathcal{H}_C$.
In each step, the walker has three choices - it can move to the left,
move to the right or just stay at the current position.
To each of these options, we assign a vector of the standard basis of the coin space $\mathcal{H}_C$, i.e. the coin space is three dimensional
$$\mathcal{H}_C = \mathbb{C}^3 = \textrm{Span}\{\ket{L}, \ket{S}, \ket{R}\}.   $$
The evolution operator realizing a single step of the three state quantum walk is given by
$  U = S \cdot (I_P \otimes C)$
where $S$ is the position shift operator, $I_P$ is the identity operator of the position space
$\mathcal{H}_P$ and $C$ is the coin flip operator.
In the three state quantum walk on a line, the position shift operator $S$ has the form of
$$
S = \sum_{m = -\infty}^{+\infty}
                \ket{m-1}\bra{m} \otimes \ket{L}\bra{L} +
                \ket{m}\bra{m} \otimes \ket{S}\bra{S} +
                \ket{m+1}\bra{m} \otimes \ket{R}\bra{R}.
$$
As for the coin operator $C$, a common choice is the Grover operator $G$.
The Grover operator is originally designed for
Grover's search algorithm~\cite{grover1997quantum},
and now finds its use in quantum walks. The Grover operator is defined as
\begin{equation}
    \operator{G} = \frac{1}{3}
        \begin{pmatrix}
            -1  &   2   &   2   \\
            2   &   -1  &   2   \\
            2   &   2   &   -1
        \end{pmatrix}.
\end{equation}

The state of the walker after evolving $t$ steps is given by the successive applications of
the evolution operator $U$ on the initial state.
Let $\ket{\psi(t)}$ be the system state after walking $t$ steps, then
\begin{equation}
    \ket{\psi(t)} = \sum_{m} \ket{m} \otimes \big( \psi_L(t,m)\ket{L}  +
                    \psi_S(t,m)\ket{S} + \psi_R(t,m)\ket{R} \big)
                    = U^t\ket{\psi(0)},
\end{equation}
where $\ket{\psi(0)}$ is the initial state, $\psi_L(t,m)$ is the
probability amplitude of the walker being at position $m$ with coin state $\ket{L}$
after walking $t$ steps. $\psi_S(t,m)$ and $\psi_R(t,m)$ are defined similarly.
We will write $\ket{m, L}$, $\ket{m, S}$ and $\ket{m, R}$ for short of
$\ket{m}\otimes\ket{L}$, $\ket{m}\otimes\ket{S}$ and $\ket{m}\otimes\ket{R}$
whenever there is no ambiguity.
Let $P(t,m)$ be the probability of finding the walker at position $m$ after
walking $t$ steps, then
$$
P(t,m)  = \vert \psi_L(t,m) \vert^2 + \vert \psi_S(t,m) \vert^2 + \vert \psi_R(t,m) \vert^2.
$$

In summary, the process of Grover walks on a line can be described as follows.
\begin{description}
  \item[Step1.] Initialize the system state to
      $\ket{\psi_{0}} = \alpha\ket{0, L} + \beta\ket{0, S} + \gamma\ket{0, R} $,
      where $\alpha, \beta, \gamma \in \mathbb{C}$ and
      $\vert\alpha\vert^2 + \vert\beta\vert^2 + \vert\gamma\vert^2 = 1$.
  \item[Step2.] For any chosen number of steps $t$, apply $U$ to the system $t$ times.
  \item[Step3.] Measure the system state $\ket{\psi(t)}$ to get the walker's position
                probability distribution.
\end{description}

\subsection{Grover walks with one boundary}
\label{sec:one-boundary}
In this process, we introduce an absorbing boundary into the line,
resulting in Grover walks on a semi-infinite line.
This can be done by setting a measurement device
which corresponds to answering the question \textit{"Is the walker at position $n$?"}.
The measurement is implemented as two projection operators
$$
\Pi_{yes}^{n} = \ket{n}\bra{n} \otimes I_C, \quad  \Pi_{no}^{n} = I - \Pi_{yes}^{n},
$$
where $I_C$ is the identity operator of the coin space $\mathcal{H}_C$
and $I$ is the identity operator of the Hilbert space $\mathcal{H}$.

As an example of the measurement, suppose the system is now in state
$
\ket{\psi(t)} = \frac{1}{\sqrt{5}}\ket{0, L} + \frac{1}{\sqrt{5}}\ket{0, S} +
                \frac{1}{\sqrt{5}}\ket{0, R} +  \frac{1}{\sqrt{5}}\ket{1, S} +
                \frac{1}{\sqrt{5}}\ket{2, R},
$
and is measured by the projection operator $\Pi_{yes}^{0}$
(corresponding to the question \textit{"Is the walker at position $0$?"}).
The answer \textit{yes} is obtained with probability
\begin{eqnarray*}
  \| \Pi_{yes}^{0}\ket{\psi(t)} \|^2 &=&
        \|  \Pi_{yes}^{0}
            \big(
                        \frac{1}{\sqrt{5}}\ket{0, L} + \frac{1}{\sqrt{5}}\ket{0, S} +
                        \frac{1}{\sqrt{5}}\ket{0, R} + \frac{1}{\sqrt{5}}\ket{1, S} +
                        \frac{1}{\sqrt{5}}\ket{2, R}
            \big) \|^2      \\
  ~ &=& \|  \frac{1}{\sqrt{5}}\ket{0, L} + \frac{1}{\sqrt{5}}\ket{0, S} +
            \frac{1}{\sqrt{5}}\ket{0, R} \|^2 = \frac{3}{5},
\end{eqnarray*}
in which case the system state collapses to
$
\ket{\psi(t)}_{yes} =    \frac{1}{\sqrt{3}}\ket{0, L} +
                            \frac{1}{\sqrt{3}}\ket{0, S} + \frac{1}{\sqrt{3}}\ket{0, R}.
$
The answer is \textit{no} with probability $2/5$,
in which case the system collapses to state
$
\ket{\psi(t)}_{no} = \frac{1}{\sqrt{2}}\ket{1, S} + \frac{1}{\sqrt{2}}\ket{2, R}.
$

Analogous to Grover walks on a line,
Grover walks on a line with one boundary can be depicted as follows:
\begin{description}
  \item[Step1.] Initialize the system state to
      $\ket{\psi_{0}} = \alpha\ket{0, L} + \beta\ket{0, S} + \gamma\ket{0, R} $,
      where $\alpha, \beta, \gamma \in \mathbb{C}$ and
      $\vert\alpha\vert^2 + \vert\beta\vert^2 + \vert\gamma\vert^2 = 1$.
  \item[Step2.] For each step of the evolution
    \begin{enumerate}
      \item Apply $U$ to the system.
      \item Measure the system according to $\{ \Pi_{yes}^{-M}, \Pi_{no}^{-M} \}$
            to test whether the walker is or not at $-M \; (M > 0)$.
    \end{enumerate}
\item[Step3.] If the measurement result is \textit{yes} (i.e. the walker is at $-M$),
            then terminate the process, otherwise repeat \textbf{Step2}.
\end{description}

We are interested the probability that the measurement of whether the walker is at position $-M$
eventually results in yes, which is called the absorbing probability.
Let $P_{\underline{-M,0,\infty}}(\alpha,\beta,\gamma)$ denotes the absorbing probability,
where $-M$, $0$, and $\infty$ represent the left boundary,
the walker's initial position and the right boundary respectively,
and $\alpha,\beta,\gamma$ are the probability amplitudes of
the coin components $\ket{L}, \ket{S}, \ket{R}$.
To keep accordance with the two boundaries case in symbols,
we assume there is a right boundary which is infinitely far away in the one boundary case.
That's why we have a rather confusing $\infty$ here.

\subsection{Grover walks with two boundaries}
\label{sec:two-boundaries}
The third process is similar to Grover walks with one boundary,
except that two boundaries are presented rather than one.
Specifically, using the same measurement devices as defined for semi-infinite Grover walks,
we describe the following process which is called Grover walks on a line with two boundaries,
or the finite Grover walks.
\begin{description}
  \item[Step1.] Initialize the system state to
      $\ket{\psi_{0}} = \alpha\ket{0, L} + \beta\ket{0, S} + \gamma\ket{0, R} $,
      where $\alpha, \beta, \gamma \in \mathbb{C}$ and
      $\vert\alpha\vert^2 + \vert\beta\vert^2 + \vert\gamma\vert^2 = 1$.
  \item[Step2.] For each step of the evolution
  \begin{enumerate}
  \item Apply $U$ to the system.
  \item Measure the system according to $\{ \Pi_{yes}^{-M}, \Pi_{no}^{-M} \}$
                    to test whether the walker is or not at $-M  \; (M > 0)$.
        $-M$ is the left absorbing boundary.
  \item Measure the system according to $\{ \Pi_{yes}^{N}, \Pi_{no}^{N} \}$
                    to test whether the walker is or not at $N  \; (N > 0)$.
        $N$ is the right absorbing boundary.
  \end{enumerate}
  \item[Step3.] If either of the measurement results is \textit{yes}
                (i.e. the walker is either at $-M$ or $N$), then terminate the process,
                otherwise repeat \textbf{Step2}.
\end{description}

We are interested in the absorbing probabilities that the walker is eventually
absorbed by the left or the right boundary. Let
\begin{itemize}
    \item $P_{\underline{-M,0,N}}(\alpha,\beta,\gamma)$
                be the left absorbing probability that the measurement of
                whether the walker is at position $-M$ eventually results in yes.
    \item $Q_{\underline{-M,0,N}}(\alpha,\beta,\gamma)$
                be the right absorbing probability that the measurement of
                whether the walker is at position $N$ eventually results in yes.
\end{itemize}
In the above conventions, $-M$, $0$, and $\infty$ represent the left boundary,
the walker's initial position and the right boundary respectively,
and $\alpha,\beta,\gamma$ are the probability amplitudes of
the coin components $\ket{L}, \ket{S}, \ket{R}$.
\section{One Boundary}\label{sec:one-boundary}

We begin with three special initial cases:
1) the initial state is $\ket{0, L}$;
2) the initial state is $\ket{0, S}$;
and
3) the initial state is $\ket{0, R}$.
The boundary is fixed at $-1$ for above three cases.
We will define generating functions for these simple cases
which are used to determine absorbing probabilities for all boundaries.
The methods applied in this section are inspired by~\cite{bach2004one}.
we thank them for offering such elegant methods.

\subsection{Generating functions}
\label{sec:one-boundary-generating-function}

We first consider simple cases where the boundary is at $-1$,
and the initial state is $\ket{0, L}$, $\ket{0, S}$ and $\ket{0, R}$ respectively.
We define the following generating functions $l(z)$, $s(z)$ and $r(z)$ for each
of above three cases
\begin{eqnarray}
  l(z)  &=&     \sum_{t=1}^{\infty} \bra{-1,L}U(\Pi_{no}^{-1}U)^{t-1}\ket{0,L} z^t,
  \label{eq:l-generating-function}   \\
  s(z)  &=&     \sum_{t=1}^{\infty} \bra{-1,L}U(\Pi_{no}^{-1}U)^{t-1}\ket{0,S} z^t,
  \label{eq:s-generating-function}   \\
  r(z)  &=&     \sum_{t=1}^{\infty} \bra{-1,L}U(\Pi_{no}^{-1}U)^{t-1}\ket{0,R} z^t.
  \label{eq:r-generating-function}
\end{eqnarray}
$\bra{-1,L}U(\prod_{no}^{-1}U)^{t-1}\ket{0,L}$ is the probability amplitude
with which the walker first reaches the boundary $-1$ after walking $t$ steps
when starting with state $\ket{0, L}$.
It's easy to see that we encode all the probability amplitudes
that lead the walker to $-1$ into the coefficients of $z^t$ in $l(z)$.
$s(z)$ and $r(z)$ are similarly defined except that the system is initialized to
$\ket{0,S}$ and $\ket{0,R}$ respectively.

Recall that we denote
$P_{\underline{-1,0,\infty}}(1,0,0)$,
$P_{\underline{-1,0,\infty}}(0,1,0)$ and
$P_{\underline{-1,0,\infty}}(0,0,1)$
as the probabilities that a walker starting in state $\ket{0, L}, \ket{0, S}$ or $\ket{0, R}$
is eventually absorbed by the boundary $-1$.
These probabilities can be calculated by summing up the squared amplitudes
encoded in the generating functions:
\begin{eqnarray*}
  P_{\underline{-1,0,\infty}}(1,0,0)    &=&     \sum_{t=1}^{\infty} \Big\| [z^t]l(z) \Big\|^2, \\
  P_{\underline{-1,0,\infty}}(0,1,0)    &=&     \sum_{t=1}^{\infty} \Big\| [z^t]s(z) \Big\|^2, \\
  P_{\underline{-1,0,\infty}}(0,0,1)    &=&     \sum_{t=1}^{\infty} \Big\| [z^t]r(z) \Big\|^2,
\end{eqnarray*}
where $[z^t]l(z)$ is the coefficient of $z^t$ in $l(z)$, and similarly for $[z^t]s(z)$, $[z^t]r(z)$.

Given two arbitrary generating functions $u$ and $v$,
their Hadamard product is $u \odot v$, defined as
$$ (u \odot v)(z) = \sum_{t = 1}^{\infty} \big[([z^t]u(z))([z^t]v(z))\big]z^t.$$
Thus,
$P_{\underline{-1,0,\infty}}(1,0,0) = ( l \odot \bar{l})(1)$,
$P_{\underline{-1,0,\infty}}(0,1,0) = ( s \odot \bar{s})(1)$ and
$P_{\underline{-1,0,\infty}}(0,0,1) = ( r \odot \bar{r})(1)$.
In general we have
\begin{equation}\label{eq:hadamard-product}
  (u \odot v)(1) = \frac{1}{2\pi}\int_{0}^{2\pi} u(e^{i\theta})v(e^{-i\theta})d\theta,
\end{equation}
provided that $\sum_{t=1}^{\infty} ([z^t]u(z))([z^t]v(z))$ converges.
Let $L(\theta) = l(e^{i\theta}), S(\theta) = s(e^{i\theta}), R(\theta) = s(e^{i\theta})$,
we can calculate the absorbing probabilities in analytical form using
Equation \ref{eq:hadamard-product}:
\begin{eqnarray}
  \label{eq:l-probability}
  P_{\underline{-1,0,\infty}}(1,0,0)    &=&
        \frac{1}{2\pi}\int_{0}^{2\pi} \vert L(\theta) \vert^2 d\theta \\
  \label{eq:s-probability}
  P_{\underline{-1,0,\infty}}(0,1,0)    &=&
        \frac{1}{2\pi}\int_{0}^{2\pi} \vert S(\theta) \vert^2 d\theta \\
  \label{eq:r-probability}
  P_{\underline{-1,0,\infty}}(0,0,1)    &=&
        \frac{1}{2\pi}\int_{0}^{2\pi} \vert R(\theta) \vert^2 d\theta.
\end{eqnarray}

The generating functions defined by Equations
\ref{eq:l-generating-function}-\ref{eq:r-generating-function} can be solved.
The solving procedure is detailed in Appendix \ref{appendix:derive-recurrences}.
We present the results in Theorem \ref{thm:one-boundary}.
\begin{theorem}\label{thm:one-boundary}
  The generating functions $l(z)$, $s(z)$ and $r(z)$
  defined in Grover walks with one boundary satisfy the following recurrences
  \begin{eqnarray}
    l(z)  &=& \frac{-3 - 4z - 3z^2 + (1+z)\Delta}{2z},    \\
    s(z)  &=& \frac{-3 -z          + \Delta}{2z},   \\
    r(z)  &=& \frac{ 3 + 2z + 3z^2 + (-1+z)\Delta}{4z},
  \end{eqnarray}
  where $ \Delta = \sqrt{9 + 6z + 9z^2}$.
\end{theorem}

Then by Equations \ref{eq:l-probability}-\ref{eq:r-probability},
we can calculate the absorbing probabilities for these three special cases:
\begin{eqnarray*}
  P_{\underline{-1,0,\infty}}(1,0,0)
            &=&
            \frac{1}{2\pi}\int_{0}^{2\pi} \vert L(\theta) \vert^2 d\theta
            \;\;=\;\;
            \frac{1}{2\pi}\int_{0}^{2\pi} \vert l(e^{i\theta}) \vert^2 d\theta
            \;\;\doteq\;\;
            0.4248,  \\
  P_{\underline{-1,0,\infty}}(0,1,0)
            &=&
            \frac{1}{2\pi}\int_{0}^{2\pi} \vert S(\theta) \vert^2 d\theta
            \;\;=\;\;
            \frac{1}{2\pi}\int_{0}^{2\pi} \vert s(e^{i\theta}) \vert^2 d\theta
            \;\;\doteq\;\;
            0.5255,  \\
  P_{\underline{-1,0,\infty}}(0,0,1)
            &=&
            \frac{1}{2\pi}\int_{0}^{2\pi} \vert R(\theta) \vert^2 d\theta
            \;\;=\;\;
            \frac{1}{2\pi}\int_{0}^{2\pi} \vert r(e^{i\theta}) \vert^2 d\theta
            \;\;\doteq\;\;
            0.6693.
\end{eqnarray*}

The analytical absorbing probabilities are matched with the simulation results,
as shown in Figure \ref{fig:one-boundary-exit-probabilities}.
We can see that the absorbing probabilities converge toward their
limiting values very quickly.
The fast convergence indicates that the remaining
probability amplitudes spread rapidly to the right and never go back.
\begin{figure}[!ht]
  \begin{minipage}[t]{0.48\textwidth}
    \centering
    \includegraphics[width=1.0\textwidth]{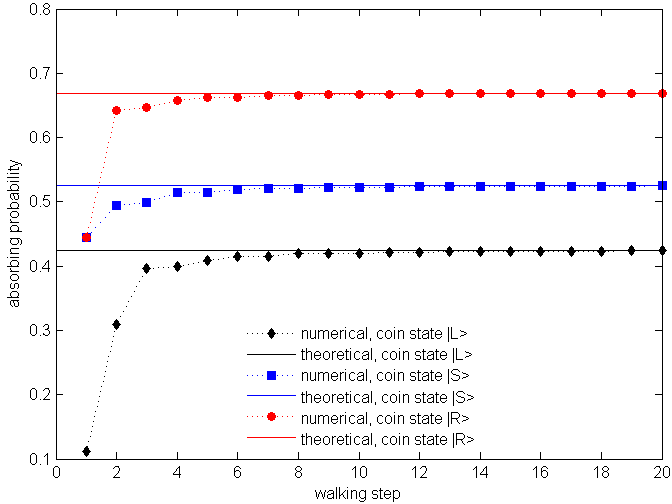}
    \caption{(Color online): Numerical simulations of the
              absorbing probabilities evolving with walking steps
              in Grover walks with one boundary.
              The boundary is fixed at $-1$.
              Solid lines represent the theoretical probabilities.
              Dashed lines represent the numerical probabilities.
            }
    \label{fig:one-boundary-exit-probabilities}
  \end{minipage}
  \hspace{0.02\textwidth}
  \begin{minipage}[t]{0.48\textwidth}
    \centering
    \includegraphics[width=1.0\textwidth]{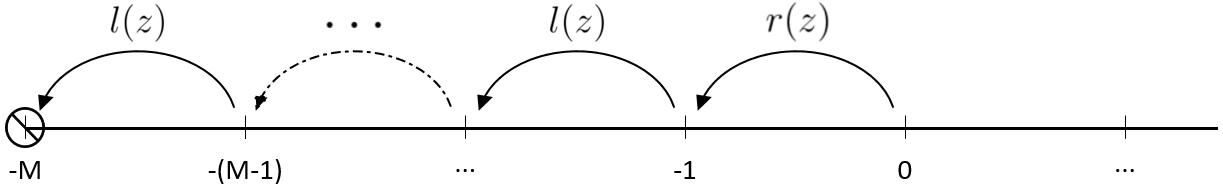}
    \caption{ The walker is initialized to state $\ket{0,R}$ and the boundary is at $-M$.
              In order to reach the boundary,
              the walker has to move left $M$ times effectively.
              For each move after the first, the coin state is $\ket{L}$.
              The first move is with coin state $\ket{R}$.}
    \label{fig:one-boundary-m}
  \end{minipage}
\end{figure}

\subsection{Arbitrary boundary}
\label{sec:one-boundary-arbitrary-boundary}
The reason that the generating functions $l(z)$, $s(z)$ and $r(z)$ are important
in analyzing absorbing probability for arbitrary boundary is as follows.
Suppose now that the boundary is located at $-M$ for some $M \geq 1$
and the walker's initial state is $\ket{0, R}$.
Consider a generating function defined similarly to $r(z)$,
except for the boundary at position $-M$ rather than $1$.
Then this generating function is simply $r(z)(l(z))^{M-1}$,
which follows from the fact that
to reach the boundary $-M$ from original position $0$,
the walker has to move left $M$ times effectively.
For each move after the first, the coin state is always $\ket{L}$.
The first move is with coin state $\ket{R}$.
The process is depicted in Figure \ref{fig:one-boundary-m}.
Likewise, the generating functions corresponding to staring in
states $\ket{0,L}$ and $\ket{0,S}$ are simply $(l(z))^M$ and $s(z)(l(z))^{M-1}$.
Then from the discussions in the previous section we
know how to calculate the absorbing probabilities for these simple cases:
\begin{eqnarray*}
  P_{\underline{-M,0,\infty}}(1,0,0)    &=&
                    \frac{1}{2\pi}\int_{0}^{2\pi}
                    \vert L(\theta) \vert^{2M} d\theta,     \\
  P_{\underline{-M,0,\infty}}(0,1,0)    &=&
                    \frac{1}{2\pi}\int_{0}^{2\pi}
                    \vert S(\theta) \vert^{2} \vert L(\theta) \vert^{2M-2} d\theta,    \\
  P_{\underline{-M,0,\infty}}(0,0,1)    &=&
                    \frac{1}{2\pi}\int_{0}^{2\pi}
                    \vert R(\theta) \vert^{2} \vert L(\theta) \vert^{2M-2} d\theta.
\end{eqnarray*}

Figure \ref{fig:one-moving-boundary} illustrates the relation between
the absorbing probabilities (on $y$-axis) and the boundary $M$ (on $x$-axis).
We can see that the absorbing probability $P_{\underline{-M,0,\infty}}(0,0,1)$
undergoes an extreme fast decay when the boundary is moved from $-1$ to $-2$,
and then rapidly reaches its limiting value as the boundary becomes large.
This rather strange phenomenon is emerged due to the
\textit{oscillating localization} effect in Grover walks.
It has been proved that
Grover walks on an infinite line will result in localization
when the system is initialized to state $\ket{0,R}$.
The localization probability is $0.202$ at the origin (position $0$)~\cite{falkner2014weak},
and is exponentially decaying with the distance from the origin~\cite{vstefavnak2014limit}.
However, localization is also shown at position $-1$ in this special case,
which violates the exponentially decaying conclusion stated in~\cite{vstefavnak2014limit},
and has not been studied to the best of my knowledge.
The localization probability at $-1$ is approximate to $0.202$ by simulation,
as shown in Figure \ref{fig:figuer04}.
From Figure \ref{fig:figuer04} we can infer that,
the probabilities of finding the walker at position $-1$ and $0$
are sum to constant $0.404$, while the two probabilities are oscillating around $0.202$.
These two probabilities will finally both be $0.202$
when the walker evolves large enough steps.
By analyzing the numerical data,
we also find that the probabilities at the left of $-1$ decay exponentially
with the distance from $-1$,
while the probabilities at the right of $0$ decay exponentially
with the distance from $0$.
This is a new \textit{two-peak localization} phenomenon.
We conjecture that it is the oscillating effect that results in the two-peak localization -
a fair part of the system state is oscillating between position $-1$ and $0$,
and never leave that region.
When a boundary is located at $-1$,
the amplitudes that should have been oscillating between $-1$ and $0$
are absorbed by that boundary, resulting in the disappearance of localization.
When a boundary is located at $-2$,
some of the amplitudes are oscillating between position $-1$ and $0$,
resulting in the amplitudes absorbed by the boundary $-2$ are much less than
in the former case.
The two-peak oscillation now revives at $-1$ and $0$.
From this point of view, we figure out the reason for the extreme fast decay
of $P_{\underline{-M,0,\infty}}(0,0,1)$.
\begin{figure}[!ht]
  \begin{minipage}[t]{0.48\textwidth}
    \centering
    \includegraphics[width=1.0\textwidth]{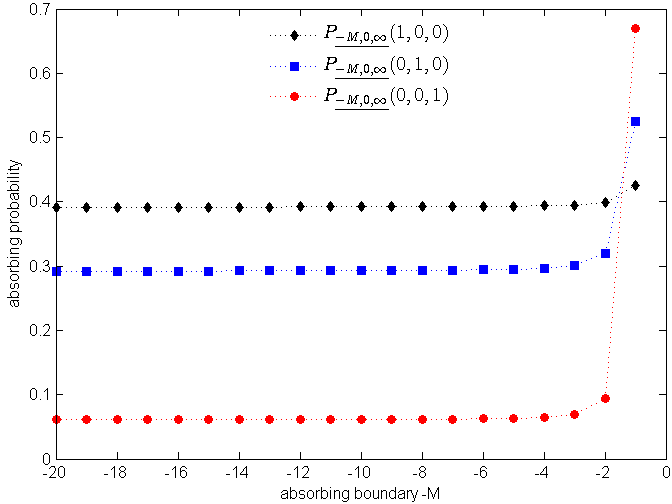}
    \caption{(Color online) The theoretical absorbing probabilities
      $P_{\underline{-M,0,\infty}}(1,0,0)$,
      $P_{\underline{-M,0,\infty}}(0,1,0)$ and
      $P_{\underline{-M,0,\infty}}(0,0,1)$
      evolving with the absorbing boundary in Grover walks with one boundary.}
    \label{fig:one-moving-boundary}
  \end{minipage}
  \hspace{0.02\textwidth}
  \begin{minipage}[t]{0.48\textwidth}
    \centering
    \includegraphics[width=1.0\textwidth]{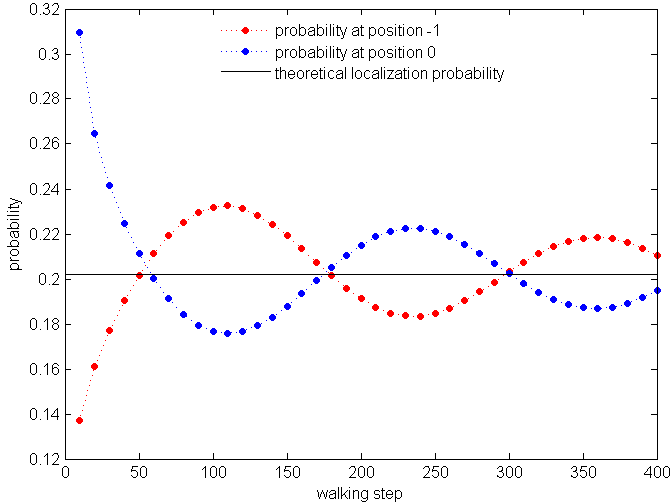}
    \caption{(Color online) The oscillating localization effect in Grover walks.
            The probabilities at position $-1$ and $0$ are oscillating around
            the theoretical localization probability $0.202$.
            The sum of these two probabilities is a constant $0.404$.}
    \label{fig:figuer04}
  \end{minipage}
\end{figure}

We have only considered the cases where the boundaries were located to
the left of the walker so far.
The symmetric cases (where the boundaries are to the right of the walker)
are easy to analyze as the Grover operator is permutation symmetric.
A Grover walk with some initial state $(\alpha,\beta,\gamma)$ and a boundary $M \geq 0$
is equivalent to
a Grover walk with some initial state $(\gamma,\beta,\alpha)$ and a boundary $-M$
in the sense that they have the symmetric probability distribution and
the same absorbing probability.
Let $P_{\underline{\infty,0,M}}(\alpha,\beta,\gamma)$ denotes the absorbing probability
where the boundary is positioned at $M \geq 0$.
Its relationship to $P_{\underline{-M,0,\infty}}(\alpha,\beta,\gamma)$
is stated in the following proposition.
\begin{proposition}
  In Grover walks with one boundary, the absorbing probability satisfies
  $$
    P_{\underline{\infty,0,M}}(\alpha,\beta,\gamma) =
    P_{\underline{-M,0,\infty}}(\gamma,\beta,\alpha),
  $$
  where $M \geq 0$ is a boundary,
  and $(\alpha,\beta,\gamma)$ are amplitudes corresponding to
  $\ket{L}$, $\ket{S}$, $\ket{R}$ coin components.
\end{proposition}
\section{Two Boundaries}\label{sec:two-boundaries}

We study Grover walks with two boundaries in the following way.
First, we fix the left boundary at $-1$,
and move the right boundary $N$ relatively to observe how the left absorbing probability
changes with $N$. To simplify the analysis, we consider three simple cases:
1) the initial state is $\ket{0, L}$;
2) the initial state is $\ket{0, S}$;
and
3) the initial state is $\ket{0, R}$.
The left boundary is fixed at $-1$ for above cases.
The left absorbing probabilities
$P_{\underline{-1,0,N}}(1,0,0)$,
$P_{\underline{-1,0,N}}(0,1,0)$, and
$P_{\underline{-1,0,N}}(0,0,1)$
are studied both numerically and analytically.
Recall that these notations are defined in Section \ref{sec:two-boundaries}.
Then, arbitrary left boundary $-M$ and arbitrary right boundary $N$ are investigated.
We make use of the generating functions defined in the former case to
express the absorbing probabilities studied in this one.
We will show that the sum of the left and right absorbing probabilities
is less than $1$ for almost arbitrary left and right boundaries $-M, N$,
which is strikingly different from that of Hadamard walks and random walks with boundaries.
This is due to the localization effect uniquely in Grover walks.
Some of the system state is trapped near the origin and cannot be absorbed by either boundary,
resulting in the sum less than $1$.

\subsection{Generating functions}

Three special initial cases are considered:
the initial state is $\ket{0, L}$, $\ket{0, S}$ and $\ket{0, R}$ respectively,
and the left boundary is always fixed at $-1$.
We define three functions $l(N, z), s(N, z), r(N, z)$ for these cases as
\begin{eqnarray}
  l(N, z)   &=&
            \sum_{t=1}^{\infty} \bra{-1,L}U(\Pi_{no}^{N}\Pi_{no}^{-1}U)^{t-1}\ket{0,L} z^t,
  \label{eq:l-generating-function-boundaries}     \\
  s(N, z)  &=&
            \sum_{t=1}^{\infty} \bra{-1,L}U(\Pi_{no}^{N}\Pi_{no}^{-1}U)^{t-1}\ket{0,S} z^t,
  \label{eq:s-generating-function-boundaries}     \\
  r(N, z)  &=&
            \sum_{t=1}^{\infty} \bra{-1,L}U(\Pi_{no}^{N}\Pi_{no}^{-1}U)^{t-1}\ket{0,R} z^t.
  \label{eq:r-generating-function-boundaries}
\end{eqnarray}
Some explanations on the left generating function $l(N, z)$:
\begin{itemize}
  \item The initial state is $\ket{0, L}$.
  \item The left boundary is fixed at $-1$.
  \item The right boundary is at $N \geq 1$.
  \item $\bra{-1,L}U(\prod_{no}^{N}\prod_{no}^{-1}U)^{t-1}\ket{0,L}$ is the (non-normalized)
              probability amplitude that the quantum walker first hits the
              left boundary $-1$ before hitting the right boundary $N$ after walking $t$ steps.
  \item We encode all probability amplitudes that will result in the left absorption
        into the coefficients of $z^t$ in $l(N, z)$.
\end{itemize}
Other two generating functions $s(N,z), r(N, z)$ have the same explanations as $l(N, z)$
except that the latter two have initial states $\ket{0, S}$, $\ket{0, R}$ respectively.

Let's define
$L(N,\theta) = l(N, e^{i\theta})$,
$S(N,\theta) = s(N, e^{i\theta})$ and
$R(N,\theta) = r(N, e^{i\theta})$.
Based on the reasoning techniques described in Section \ref{sec:one-boundary-generating-function},
we can calculate the absorbing probabilities for above three cases
by following equations
\begin{eqnarray}
  P_{\underline{-1,0,N}}(1,0,0)     &=&
            \sum_{t=1}^{\infty} \Big\| [z^t]l(N, z) \Big\|^2    \;\;=\;\;
            \frac{1}{2\pi}\int_{0}^{2\pi} \vert L(N,\theta) \vert^2 d\theta,
  \label{eq:l-probability-boundaries}        \\
  P_{\underline{-1,0,N}}(0,1,0)     &=&
            \sum_{t=1}^{\infty} \Big\| [z^t]s(N, z) \Big\|^2    \;\;=\;\;
            \frac{1}{2\pi}\int_{0}^{2\pi} \vert S(N,\theta) \vert^2 d\theta,
  \label{eq:s-probability-boundaries}        \\
  P_{\underline{-1,0,N}}(0,0,1)     &=&
            \sum_{t=1}^{\infty} \Big\| [z^t]r(N, z) \Big\|^2    \;\;=\;\;
            \frac{1}{2\pi}\int_{0}^{2\pi} \vert R(N,\theta) \vert^2 d\theta,
  \label{eq:r-probability-boundaries}
\end{eqnarray}
where $[z^t]l(N,z)$ is the coefficient of $z^t$ in $l(N,z)$,
and similarly for $[z^t]s(N,z)$, $[z^t]r(N,z)$.

In Grover walks with two boundaries, the generating functions defined in Equations
\ref{eq:l-generating-function-boundaries},
\ref{eq:s-generating-function-boundaries} and
\ref{eq:r-generating-function-boundaries} can be solved.
For details on the solving procedure,
the readers can refer to Appendix \ref{appendix:derive-recurrences-two-boundaries}.
We only summarize the results here.
Let $l(0, z) = s(0, z) = r(0, z) = 0$,
then for any $N \geq 1$, these generating functions satisfy the following recurrences
\begin{eqnarray*}
  l(N, z)   &=&     -\frac{1}{3} z + \frac{2}{3} z \cdot s(N, z) +
                    \frac{2}{3} z \cdot l(N, z) \cdot r(N-1, z),        \\
  s(N, z)  &=&      \frac{2}{3} z - \frac{1}{3} z \cdot s(N, z) +
                    \frac{2}{3} z \cdot l(N, z) \cdot r(N-1, z),        \\
  r(N, z)  &=&      \frac{2}{3} z + \frac{2}{3} z \cdot s(N, z) -
                    \frac{1}{3} z \cdot l(N, z) \cdot r(N-1, z).
\end{eqnarray*}
Assuming $r(N-1, z)$ to be constant,
we find the solutions to these equations, as stated in Theorem \ref{thm:two-boundaries}.
\begin{theorem}\label{thm:two-boundaries}
  The generating functions $l(N, z)$, $s(N, z)$ and $r(N, z)$
  defined in Grover walks with two boundaries satisfy the following recurrences
  \begin{eqnarray}
    l(N, z)   &=&     \frac{-z + z^2}
                          {3 + z - (2z + 2z^2) \cdot r(N-1,z)},
    \label{eq:l-generating-function-solution-boundaries}         \\
    s(N, z)  &=&      \frac{2z - 2 z^2 \cdot r(N-1,z)}
                          {3 + z - (2z + 2z^2) \cdot r(N-1,z)},
    \label{eq:s-generating-function-solution-boundaries}         \\
    r(N, z)  &=&      \frac{2z + 2 z^2 - (z^2 + 3z^3) \cdot r(N-1,z) }
                          {3 + z - (2z + 2z^2) \cdot r(N-1,z)},
    \label{eq:r-generating-function-solution-boundaries}
  \end{eqnarray}
  for arbitrary $N \geq 1$, with initial conditions $l(0, z) = s(0, z) = r(0, z) = 0$.
\end{theorem}
By combining Equations
\ref{eq:l-probability-boundaries}-\ref{eq:r-generating-function-solution-boundaries},
we can calculate the left absorbing probabilities
$P_{\underline{-1,0,N}}(1,0,0)$,
$P_{\underline{-1,0,N}}(0,1,0)$ and
$P_{\underline{-1,0,N}}(0,0,1)$ for arbitrary right boundary $N \geq 0$.
As an example, we calculate these values for $10$ different right boundaries.
The results are shown in Figure \ref{fig:two-boundaries-exit-probabilities-2}.
We can see that the left absorbing probability reaches to
its limiting value rapidly when the right boundary is shifting.

We derive a closed recurrence for $P_{\underline{-1,0,N}}(0,0,1)$,
as stated in Theorem \ref{thm:recurrence}.
The theorem perfectly matches with the simulation data.
We prove its correctness in Appendix \ref{appendix:proof-of-theorem-4}.
Solving the recurrence, we get the left absorbing probability when the right boundary is
infinitely far to the right:
$
\lim_{N \rightarrow \infty}P_{\underline{-1,0,N}}(0,0,1) = 1/\sqrt{2}.
$ 
\begin{theorem}\label{thm:recurrence}
  The left absorbing probabilities $P_{\underline{-1,0,N}}(0,0,1)$ obey the following recurrence.
  \begin{eqnarray*}
    P_{\underline{-1,0,0}}(0,0,1)   &=& 0, \\
    P_{\underline{-1,0,N+1}}(0,0,1) &=&
              \frac{2 + 3P_{\underline{-1,0,N}}(0,0,1)}
                    {3 + 4P_{\underline{-1,0,N}}(0,0,1)}, \quad N \geq 1.
  \end{eqnarray*}
\end{theorem}


\begin{figure}[!ht]
  \begin{minipage}[t]{0.48\textwidth}
    \centering
    \includegraphics[width=1.0\textwidth]{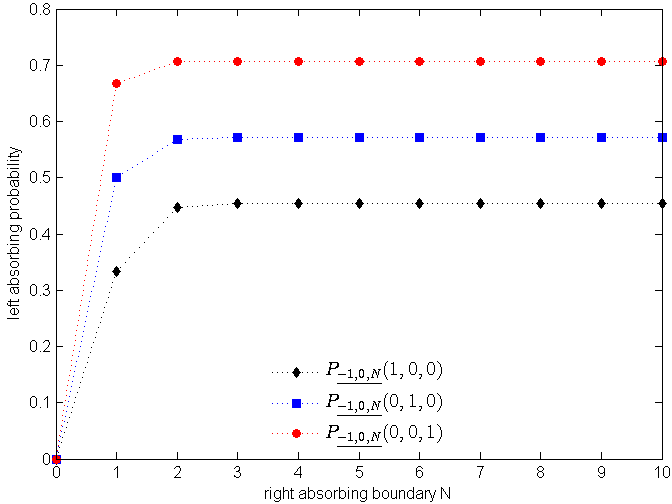}    \\
    \caption{(Color online) The theoretical left absorbing probabilities
              $P_{\underline{-1,0,N}}(1,0,0)$,
              $P_{\underline{-1,0,N}}(0,1,0)$ and
              $P_{\underline{-1,0,N}}(0,0,1)$
              evolving with the right absorbing boundary
              in Grover walks with two boundaries.
              The left absorbing boundary is at $-1$.}
    \label{fig:two-boundaries-exit-probabilities-2}
  \end{minipage}
  \hspace{0.02\textwidth}
  \begin{minipage}[t]{0.48\textwidth}
    \centering
    \includegraphics[width=1.0\textwidth]{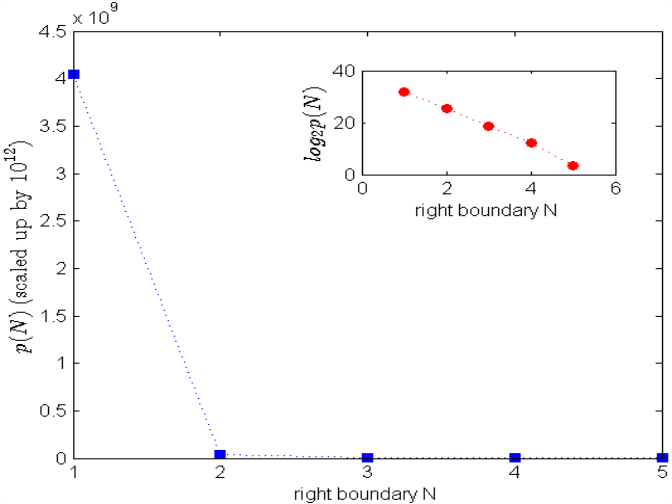}    \\
    \caption{ The localization probabilities $p(N)$
              in Grover walks with two boundaries.
              In order to make the values observable,
              we scale up $p(N)$ by $10^{12}$.
              The values of $\log_2(p(N))$ are plotted in the inset
              and decrease linearly with $N$.
              }
    \label{fig:figure06}
  \end{minipage}
\end{figure}

\subsection{Arbitrary boundaries}
Suppose now the left boundary is at $-M$ $(M \geq 1)$,
and the right boundary is at $N$ $(N \geq 1)$.
Let's define generating functions $l(-M, N, z)$, $s(-M, N, z)$ and $r(-M, N, z)$
similarly to $l(N, z)$, $s(N, z)$ and $r(N, z)$.
We now show that $r(-M, N, z)$ can be represented by $l(N, z)$ and $r(N, z)$.
For a walker with initial state $\ket{0,R}$,
if it wishes to be absorbed by the left boundary $-M$ rather than the right boundary $N$,
it must reach positions $-1, -2, \cdots, -M$ sequentially without being trapped by the right.
That is, the walker has to move left $M$ times effectively.
For the first move from $0$ to $-1$ with coin state $\ket{R}$,
$r(N, z)$ counts the paths.
For an intermediate move from $-k$ to $-(k+1)$ with coin state $\ket{L}$,
$l(N+k, z)$ counts the paths.
Then we have $r(-M, N, z) = r(N,z)\prod_{k=1}^{M-1} l(N+k,z)$.
Likewise, $l(-M, N, z) = l(N,z)\prod_{k=1}^{M-1} l(N+k,z)$
and $s(-M, N, z) = s(N,z)\prod_{k=1}^{M-1} l(N+k,z)$.
As $l(N,z)$, $s(N,z)$ and $r(N,z)$ can be calculated by the recurrences
stated in Theorem \ref{thm:two-boundaries},
$l(-M, N, z)$, $s(-M, N, z)$ and $r(-M, N, z)$ are solvable.
With $l(-M, N, z)$, $s(-M, N, z)$ and $r(-M, N, z)$ in hand,
we can calculate the left absorbing probabilities for
arbitrary left and right boundaries by the following formulas.
\begin{eqnarray}
  P_{\underline{-M,0,N}}(1,0,0)     &=&
    \sum_{t=1}^{\infty} \Big\| [z^t]l(-M, N, z) \Big\|^2    \;\;=\;\;
\frac{1}{2\pi}\int_{0}^{2\pi}\vert l(N,\theta)\prod_{k=1}^{M-1}l(N+k,\theta) \vert^2 d\theta,
    \label{eq:two-boundaries-left}     \\
  P_{\underline{-M,0,N}}(0,1,0)     &=&
    \sum_{t=1}^{\infty} \Big\| [z^t]s(-M, N, z) \Big\|^2    \;\;=\;\;
\frac{1}{2\pi}\int_{0}^{2\pi}\vert s(N,\theta)\prod_{k=1}^{M-1}l(N+k,\theta) \vert^2 d\theta,
    \label{eq:two-boundaries-stay}     \\
  P_{\underline{-M,0,N}}(0,0,1)     &=&
    \sum_{t=1}^{\infty} \Big\| [z^t]r(-M, N, z) \Big\|^2    \;\;=\;\;
\frac{1}{2\pi}\int_{0}^{2\pi}\vert r(N,\theta)\prod_{k=1}^{M-1} l(N+k,\theta) \vert^2 d\theta.
    \label{eq:two-boundaries-right}
\end{eqnarray}

Up to now, we have only discussed how to estimate the left absorbing probability $P_{\underline{-M,0,N}}(\alpha,\beta,\gamma)$.
How to calculate the right absorbing probability $Q_{\underline{-M,0,N}}(\alpha,\beta,\gamma)$?
As the Grover operator is permutation symmetric,
a Grover walk with some initial state $(\alpha,\beta,\gamma)$, a left boundary $-M$
and a right boundary $N$
is equivalent to
a Grover walk with some initial state $(\gamma,\beta,\alpha)$, a left boundary $-N$
and a right boundary $M$
in the sense that they have the symmetric probability distributions and
the symmetric left/right absorbing probabilities.
The relationship between the left and right absorbing probabilities is stated
in Proposition \ref{pro:two-boundaries}.
\begin{proposition}\label{pro:two-boundaries}
  In Grover walks with two boundaries, the left and right absorbing probabilities satisfy
  $$
  Q_{\underline{-M,0,N}}(\alpha,\beta,\gamma) =
  P_{\underline{-N,0,M}}(\gamma,\beta,\alpha),
  $$
  where $M \geq 1$, $N \geq 1$ are boundaries,
  and $(\alpha,\beta,\gamma)$ are amplitudes corresponding to
  $\ket{L}$, $\ket{S}$ $\ket{R}$ coin components.
\end{proposition}

Now we analyze how localization affects
the absorbing probabilities in Grover walks with two boundaries.
To simplify the analysis, we only consider a special case that
the system is initialized to $\ket{0, R}$.
The methods can be applied to more complicated cases smoothly.
Let's define
$l(-M, N) = P_{\underline{-M,0,N}}(0,0,1)$,
$r(-M, N) = Q_{\underline{-M,0,N}}(0,0,1)$.
$l(-M, N)$ and $r(-M, N)$ are the absorbing probabilities
that the walker starts from state $\ket{0, R}$ and is absorbed by the left and the right.
We also define $s(-M, N)$ to be the sum of two absorbing probabilities:
$ s(-M, N) = l(-M, N) + r(-M, N) $.
Table \ref{tab:absorbing-probabilities} shows the values
of $l(-2, N)$, $r(-2, N)$ and $s(-2, N)$ for different right boundaries $N$
($N$ varies from $1$ to $6$).
From Table \ref{tab:absorbing-probabilities},
we observe that $s(-2, N)$ is less than $1$,
which is in sharp contrast to the Hadamard walk case.
In Hadamard walks with two boundaries, it is stated that
$$
\forall M, N \geq 1, \quad s(-M, N) = 1,
\quad (\textrm{\cite{ambainis2001one}, PROPOSITION 9}),
$$
which indicates that the Hadamard walker is absorbed by either the left
or the right boundary after walking enough steps.
However, This equation doesn't hold in Grover walks with two boundaries.
Due to the localization effect, there is some probability that
the walker is trapped around origin.
Let $o(-M,N)$ be the localization probabilities when two boundaries are presented,
then we have a similar equation for Grover walks with two boundaries
$$
\forall M, N \geq 1, \quad s(-M,N) + o(-M,N) = 1.
$$
When the gap between the left and right boundaries becomes larger (i.e., $M+N$ becomes larger),
more localization probabilities are introduced, as these
bounded positions can reserve more probability amplitudes.
Thus the bigger the gap, the larger $o(-M,N)$ and the smaller $s(-M,N)$.

The localization probabilities are exponentially decay in Grover walks,
as stated in~\cite{inui2005one, falkner2014weak}.
This phenomenon is also observed in Grover walk with two boundaries.
Now let's dig deeper on the data given in Table \ref{tab:absorbing-probabilities}.
For arbitrary $N \geq 1$, we define $p(N) = s(-2,N) - s(-2,N+1)$.
Actually, $p(N)$ is the localization probability at position $N$
when the left boundary is at $-2$ and the right boundary is at $N+1$.
As the localization probabilities are rather small,
we scale up $p(N)$ by $10^{12}$ to observe their trends.
In order to show that $p(N)$ is truly exponentially decaying,
we also calculate the $log$ values of $p(N)$.
The values of $p(N)$ and $\log_2(p(N))$ are calculated
in Table \ref{tab:absorbing-probabilities},
and visualized in Figure \ref{fig:figure06}.
From the inset of Figure \ref{fig:figure06}, we can see that
the $log$ values of $p(N)$ decrease linearly with $N$,
which indicates the exponential decay of localization probabilities
in Grover walk with two boundaries.
\renewcommand\arraystretch{1.4}
\begin{table}[!hbt]
  \centering
  \begin{tabular}{|C{1cm}|C{1cm}|C{2.5cm}|C{2.5cm}|C{2.5cm}|C{3.5cm}|C{2cm}|}
  \hline
    $-M$ & $N$  & $l(-M,N)$
                & $r(-M,N)$
                & $s(-M,N)$  & $p(N)$ (scaled by $10^{12}$) & $\log_2(p(N))$ \\
  \hline
    -2 & 1 & 0.1529411765 & 0.4470588235  & 0.6000000000  & 4040404040  & 31.9119     \\
  \hline
    -2 & 2 & 0.1616161616 & 0.4343434343  & 0.5959595960  & 41228613    & 25.2971    \\
  \hline
    -2 & 3 & 0.1619106568 & 0.4340077105  & 0.5959183673  & 420743      & 18.6826    \\
  \hline
    -2 & 4 & 0.1619197226 & 0.4339982240  & 0.5959179466  & 4292        & 12.0674    \\
  \hline
    -2 & 5 & 0.1619199936 & 0.4339979488  & 0.5959179423  & 11          & 3.4594    \\
  \hline
    -2 & 6 & 0.1619200016 & 0.4339979407  & 0.5959179423  & ~           & ~     \\
  \hline
  \end{tabular}
  \caption{The left and right absorbing probabilities when the left boundary
          is at $-2$ and the right boundary varies from $1$ to $6$.
          These probabilities are calculated by Equations
          \ref{eq:two-boundaries-left}-\ref{eq:two-boundaries-right},
          with 10 digits reserved. $s(-M,N)$ is the sum of $l(-M,N)$ and $r(-M,N)$.
          The localization probabilities $p(N)$ are scaled up by $10^{12}$ to
          make them easy to observe. The values of $\log_2(p(N))$ decrease linearly
          with the right boundary $N$.
          }
  \label{tab:absorbing-probabilities}
\end{table}

\section{Conclusion}\label{sec:conclusion}

We analyze in detail the dynamics of Grover walks on a line with one
and two absorbing boundaries in this paper.
Both cases illustrate interesting differences between Grover walks
and Hadamard walks with boundaries.

In the one boundary case, we begin with three special initial states
and define generating functions for these simple cases.
These generating functions have closed form solutions.
Then, we use the solutions to calculate the absorbing probability for arbitrary boundary.
The oscillating localization phenomenon is observed and numerically studied
in Grover walks with one boundary.
It offers a nice explanation for the extreme fast decrease of the
absorbing probabilities when the boundary is moved from $-1$ to $-2$.

We study the two boundaries in almost the same way as what we did in the former case.
Generating functions are defined for three special initial states.
We then derive recurrence solutions to these generating functions.
These solutions are used to solve more complicated cases.
The absorbing probabilities (both the left and right absorbing probability) for
arbitrary left and right boundary can be calculated to arbitrary accuracy
by recursively applying these solutions.
When the left boundary is at $-2$, the quantum walk leads to localization.
We show that the localization probabilities decay exponentially.

Many questions are still left unsolved.
We cannot calculate the absorbing probability
when the boundary approaches infinity in the one boundary case.
In the two boundaries case,
we are unable to analytically calculate the left and right absorbing probabilities
with arbitrary coin states.
A further detailed study on these questions will appear in our forth-coming paper.

\section*{Acknowledgements}

The authors want to thank Malin Zhong, Yanfei Bai, Xiaohui Tian and Qunyong
Zhang for the insightful discussions.
This work are supported by the National Natural Science Foundation of China
(Grant Nos. 61300050, 91321312, 61321491),
the Chinese National Natural Science Foundation of Innovation Team (Grant No. 61321491),
the Research Foundation for the Doctoral Program of Higher Education of China
(Grant No. 20120091120008) and
the Science, Mathematics, and Research for Transformation (SMART) fellowship program.



\begin{appendices}

\section{Recurrences for Generating Functions of One Boundary}
\label{appendix:derive-recurrences}

In Grover walk with one boundary, the left generating function $l(z)$ in Equation
\ref{eq:l-generating-function} can be rewrote by $l(z)$, $s(z)$ and $r(z)$.
The derivation is as follows.
It must be pointed out the derivation is capable of solving three state quantum
walks with arbitrary coin operators, not limited to the Grover operator.
\begin{eqnarray*}
  l(z)  &=& \sum_{t=1}^{\infty} \bra{-1,L}U(\Pi_{no}^{-1}U)^{t-1}\ket{0,L} z^t \\
  ~     &=& \bra{-1,L}U\ket{0,L}z +
        \sum_{t=2}^{\infty} \bra{-1,L}U(\Pi_{no}^{-1}U)^{t-2}\Pi_{no}^{-1}U\ket{0,L} z^t  \\
  ~     &=& - \frac{1}{3}z +
      \frac{2}{3}z\sum_{t=2}^{\infty} \bra{-1,L}U(\Pi_{no}^{-1}U)^{t-2}\ket{0,S} z^{t-1} +
      \frac{2}{3}z\sum_{t=2}^{\infty} \bra{-1,L}U(\Pi_{no}^{-1}U)^{t-2}\ket{1,R} z^{t-1} \\
  ~     &=& - \frac{1}{3}z +
      \frac{2}{3}z\sum_{t=1}^{\infty} \bra{-1,L}U(\Pi_{no}^{-1}U)^{t-1}\ket{0,S} z^t +
      \frac{2}{3}z\sum_{t=1}^{\infty} \bra{-1,L}U(\Pi_{no}^{-1}U)^{t-1}\ket{1,R} z^t \\
  ~     &=& - \frac{1}{3}z + \frac{2}{3}z \cdot s(z) + \frac{2}{3}z \cdot l(z)r(z).
\end{eqnarray*}
By similar arguments, we get three recurrences for these generating functions:
\begin{eqnarray*}
  l(z) &=&  - \frac{1}{3}z + \frac{2}{3}z \cdot s(z) + \frac{2}{3}z \cdot l(z)r(z),  \\
  s(z) &=&    \frac{2}{3}z - \frac{1}{3}z \cdot s(z) + \frac{2}{3}z \cdot l(z)r(z),  \\
  r(z) &=&    \frac{2}{3}z + \frac{2}{3}z \cdot s(z) - \frac{1}{3}z \cdot l(z)r(z).
\end{eqnarray*}
Solving these equations and discarding the solutions that don't have Taylor expansions,
we get the desired answers
\begin{eqnarray*}
  l(z)  &=& \frac{-3 - 4z - 3z^2 + (1+z)\Delta}{2z},    \\
  s(z)  &=& \frac{-3 -z          + \Delta}{2z},   \\
  r(z)  &=& \frac{ 3 + 2z + 3z^2 + (z-1)\Delta}{4z},
\end{eqnarray*}
where $ \Delta = \sqrt{9 + 6z + 9z^2}$.

\section{Recurrences for Generating Functions of Two Boundaries}
\label{appendix:derive-recurrences-two-boundaries}

As in the case of one boundary, the left generating function $l(N, z)$ defined by
Equation \ref{eq:l-generating-function-boundaries} in Grover walk with two boundaries,
can also be represented by $l(N, z)$, $s(N, z)$ and $r(, z)$.
The derivation is as follows.
\begin{eqnarray*}
  l(N, z)   &=&
    \sum_{t=1}^{\infty} \bra{-1,L}U(\Pi_{no}^{N}\Pi_{no}^{-1}U)^{t-1}\ket{0,L} z^t  \\
  ~         &=&
    \bra{-1,L}U\ket{0,L}z +
    \sum_{t=2}^{\infty} \bra{-1,L}U (\Pi_{no}^{N}\Pi_{no}^{-1}U)^{t-2}
                                    (\Pi_{no}^{N}\Pi_{no}^{-1}U)\ket{0,L} z^t.
\end{eqnarray*}
And for any $N \geq 2$, we have
\begin{eqnarray*}
  \bra{-1,L}U\ket{0,L} &=&
        \bra{-1,L}(-\frac{1}{3}\ket{-1,L} + \frac{2}{3}\ket{0,S} + \frac{2}{3}\ket{1,R})
        = -\frac{1}{3},    \\
  \Pi_{no}^{N}\Pi_{no}^{-1}U\ket{0,L} &=&
        \frac{2}{3}\ket{0,S} + \frac{2}{3}\ket{1,R}.
\end{eqnarray*}
Then
\begin{eqnarray*}
  l(N, z)   &=&
    \bra{-1,L}U\ket{0,L}z +
    \sum_{t=2}^{\infty} \bra{-1,L}U (\Pi_{no}^{N}\Pi_{no}^{-1}U)^{t-2}
                                    (\Pi_{no}^{N}\Pi_{no}^{-1}U)\ket{0,L} z^t \\
  ~         &=&
    -\frac{1}{3}z +
    \sum_{t=2}^{\infty} \bra{-1,L}U (\Pi_{no}^{N}\Pi_{no}^{-1}U)^{t-2}
        (\frac{2}{3}\ket{0,S} + \frac{2}{3}\ket{1,R}) z^t               \\
  ~         &=&
    -\frac{1}{3}z +
  \frac{2}{3}z\sum_{t=2}^{\infty} \bra{-1,L}U (\Pi_{no}^{N}\Pi_{no}^{-1}U)^{t-2}\ket{0,S}z^{t-1} +
  \frac{2}{3}z\sum_{t=2}^{\infty} \bra{-1,L}U (\Pi_{no}^{N}\Pi_{no}^{-1}U)^{t-2}\ket{1,R}z^{t-1} \\
  ~         &=&
    -\frac{1}{3}z +
  \frac{2}{3}z\sum_{t=1}^{\infty} \bra{-1,L}U (\Pi_{no}^{N}\Pi_{no}^{-1}U)^{t-1}\ket{0,S}z^t +
  \frac{2}{3}z\sum_{t=1}^{\infty} \bra{-1,L}U (\Pi_{no}^{N}\Pi_{no}^{-1}U)^{t-1}\ket{1,R}z^t.
\end{eqnarray*}
Let's take a closer look at the last two terms.
We have
$$
\sum_{t=1}^{\infty} \bra{-1,L}U (\Pi_{no}^{N}\Pi_{no}^{-1}U)^{t-1}\ket{0,S}z^t = s(N, z)
$$
by definition.
For a walker with initial state $\ket{1, R}$, if it wants to exit from the left boundary $-1$,
it must first reach $0$.
The generating function for reaching $0$ from state $\ket{1, R}$ is $r(N-1, z)$,
while the generating function for reaching $-1$ from state $\ket{0, L}$ is $l(N,z)$.
In a word, we can derive
$$
\sum_{t=1}^{\infty} \bra{-1,L}U (\Pi_{no}^{N}\Pi_{no}^{-1}U)^{t-1}\ket{1,R}z^t
 = l(N,z)r(N-1, z).
$$
Thus
\begin{eqnarray*}
  l(N, z)   &=&
    -\frac{1}{3}z +
  \frac{2}{3}z\sum_{t=1}^{\infty} \bra{-1,L}U (\Pi_{no}^{N}\Pi_{no}^{-1}U)^{t-1}\ket{0,S}z^t +
  \frac{2}{3}z\sum_{t=1}^{\infty} \bra{-1,L}U (\Pi_{no}^{N}\Pi_{no}^{-1}U)^{t-1}\ket{1,R}z^t  \\
  ~         &=&
    -\frac{1}{3}z + \frac{2}{3}z\cdot s(N, z) + \frac{2}{3}z \cdot l(N,z)r(N-1, z).
\end{eqnarray*}
With similar derivation methods, we can get the recurrence equations
for the generating functions $s(N, z)$, $r(N, z)$.

To summarize, we derive three recurrence equations for these generating
functions defined in Grover walk with two boundaries case as
\begin{eqnarray*}
  l(N, z)   &=&     -\frac{1}{3} z + \frac{2}{3} z \cdot s(N, z) +
                    \frac{2}{3} z \cdot l(N, z)r(N-1, z)        \\
  s(N, z)  &=&      \frac{2}{3} z - \frac{1}{3} z \cdot s(N, z) +
                    \frac{2}{3} z \cdot l(N, z)r(N-1, z)        \\
  r(N, z)  &=&      \frac{2}{3} z + \frac{2}{3} z \cdot s(N, z) -
                    \frac{1}{3} z \cdot l(N, z)r(N-1, z)
\end{eqnarray*}
where the initial conditions are $l(0, z) = s(0, z) = r(0, z) = 0$ and $N \geq 0$.

\section{Proof of Theorem 4}
\label{appendix:proof-of-theorem-4}

\def \C {\mathbb C}
\def \cbar {{\overline {\mathbb C}}}
\def \R {\mathbb R}
\def \N {\mathbb N}
\def \Z{\mathbb Z}
\def \Q {\mathbb Q}
\def \D {\mathbb D}
\def \T {\mathbb T}
\def \No {{\mathbb N}_0}

This proof closely follows the argument
laid out in \cite{bach2009absorption} for the proof of Conjecture 11 in \cite{ambainis2001one}.
For the remainder of this section, let
$p_n := P_{\underline{-1,0,N}}(0,0,1)$ and
$r_n(z) := r(N,z)$. Recall the recursion governing $r_n(z)$:
$$
r_{n+1}(z)=\frac{2z(z+1)-z^2(1+3z)r_n(z)}{(z+3)-2z(z+1)r_n(z)},\hspace{0.5cm}r_0(z)=0.
$$

Prior to proving Theorem 4, we provide a few preliminary results:
\begin{proposition}\label{prop:analytical-function}
For $|w|,|z|<1$, we have $|f(w,z)|<1$ where
$$
f(w,z) = \frac{2z(z+1)-z^2(1+3z)w}{(z+3)-2z(z+1)w}.
$$
\end{proposition}
\begin{proof}
Let us rewrite $f(w,z)$ as
$ f(w,z)=\frac{1-h(z)w}{h\left(\frac{1}{z}\right) -w} $
where $h(z)=\frac{z^2(1+3z)}{2z(z+1)}$. Notice that for fixed $z$, the function $g_z(w)=f(w,z)$ is a linear fractional transformation. These transformations map circles to circles. In particular, if $|b|>1$, the map $w\mapsto\frac{1-aw}{b-w}$ maps the unit disk to another disk with center $\frac{\bar{b}-a}{|b|^2-1}$ and radius $\frac{|1-ab|}{|b|^2-1}$. If we want this transformation to map the unit disk back into the unit disk, we require the following inequality to hold:
$$
|\bar{b}-a|+|1-ab|\le |b|^2-1.
$$
Letting $a=h(z)$ and $b=h\left(\frac{1}{z}\right)$, we want to prove
$$
f_1(z):=|z|^2|z-1|^2+2|z|(1-|z|^2)|3|z|^2+3z+\bar{z}+3|\le |z+3|^2-|2z(z+1)|^2=:f_2(z)
$$
for $|z|<1$. Notice that $|b|>1$ in this case. Clearly we have $f_1(z)\le\frac{f_1(z)}{|z|}$, so $\left[ f_3(z):=\frac{f_1(z)}{|z|}\le f_2(z)\right]\Rightarrow\left[ f_1(z)\le f_2(z)\right]$. By triangle inequality and by $|z|\le 1$, we have 
$$
|3|z|^2+3z+\bar{z}+3|\le |z+1||z+3|+2|z|^2\le |z+3||z+1|+2|z|.
$$ If we let $f_4(z):=|z||z-1|^2+2(1-|z|^2)(|z+1||z+3|+2|z|)$, then $[f_4(z)\le f_2(z)]\Rightarrow [f_1(z)\le f_2(z)]$. The inequality $f_4(z)\le f_2(z)$ is equivalent to:
$$
2(1-|z|^2)|z+1||z+3|\le f_2(z)-|z||z-1|^2-2|z|(1-|z|^2).
$$
If $a\in\R$, $x:=|z|$, and $r:=\text{Re }(z)$, then we have $|z+a|^2=x^2+a^2+2ar$. By squaring both sides of the previous inequality and collecting terms, we can show that $f_4(z)\le f_2(z)$ for $|z|\le 1$ is equivalent to 
$$
4(4x^2-3)r^2+4(8x^4-13x^2+9)r+(12x^6-11x^4-42x^2+45)\ge 0,
$$
where $0\le x\le 1$ and $|r|<x$. Since the left hand side is quadratic in $r$, it is easy to prove this to be true. The result follows. 
\end{proof}
\begin{corollary}\label{cor:analytical-function}
For all $n\in\N$, the function $r_n(z)$ is analytic in $|z|<1$.
\end{corollary}
\begin{proposition}\label{prop:rewrite-rz}
For all $n\in N$, we may write $r_n(z)=\frac{p(z)}{q(z)}$
where $\textnormal{deg}(p)=\textnormal{deg}(q)+1$.
\end{proposition}
\begin{proof}
Clearly $r_1(z)$ satisfies this condition.
Let $r_{n+1}(z)=\frac{p'(z)}{q'(z)}$
and suppose $r_n(z)=\frac{p(z)}{q(z)}$
where $\text{deg}(q)=d$ and $\text{deg}(p)=\text{deg}(q)+1$.
Then $p'(z)=2z(z+1)q(z)-z^2(1+3z)p(z)$ and $q'(z)=(z+3)q(z)-2z(z+1)p(z)$.
This implies $\text{deg}(p')=d+4$ and $\text{deg}(q')=d+3$.
\end{proof}
\begin{proposition}
The functions $r_n(z)$ have the following closed form:
$$  r_n(z)=\frac{2z(z+1)R_n(z)}{R_{n+1}(z)+z^2(1+3z)R_n(z)} $$
where $R_n(z)=\lambda _+^n(z)-\lambda _-^n(z)$ and
$$
\lambda _\pm =\frac{1}{2}\left[ (3+z)-(z^2+3z^3)\pm\sqrt{(3+z+z^2+3z^3)^2-4(2z+2z^2)^2}\right].
$$
\end{proposition}
\begin{proof}
If we let $r_{n}(z)=\frac{p_n(z)}{q_n(z)}$ (not to be confused with the absorption probability $p_n$), then we have the following relation:
$$\frac{p_{n+1}(z)}{q_{n+1}(z)}=\frac{2z(z+1)q_n(z)-z^2(1+3z)p_n(z)}{(z+3)q_n(z)-2z(z+1)p_n(z)}.$$
As such, if we let $v_n(z)=\begin{bmatrix} p_n(z) \\ q_n(z)\end{bmatrix}$,
we have that $v_n(z)=(M(z))^nv_0(z)$ where $v_0=\begin{bmatrix} 0 \\ 1\end{bmatrix}$ and
$$M(z)=\begin{bmatrix} -z^2(1+3z) & 2z(z+1) \\ -2z(z+1) & z+3\end{bmatrix}.$$
We compute $M^n$ using an eigenvalue expansion:
\begin{align*}
M^nv_0 &= \frac{1}{d}\begin{bmatrix} 2z(z+1) & 2z(z+1) \\ \lambda _++z^2(1+3z) & \lambda _-+z^2(1+3z)\end{bmatrix}\begin{bmatrix} \lambda _+^n & 0 \\ 0 & \lambda _-^n\end{bmatrix}\begin{bmatrix} \lambda _-+z^2(1+3z) & -2z(z+1) \\ -\lambda _+-z^2(1+3z) & 2z(z+1)\end{bmatrix}\begin{bmatrix}0 \\ 1\end{bmatrix} \\
	&= \frac{1}{d}\begin{bmatrix} 2z(z+1) & 2z(z+1) \\ \lambda _++z^2(1+3z) & \lambda _-+z^2(1+3z)\end{bmatrix}\begin{bmatrix} \lambda _+^n & 0 \\ 0 & \lambda _-^n\end{bmatrix}\begin{bmatrix} -2z(z+1) \\ 2z(z+1)\end{bmatrix} \\
	&= \frac{1}{d}\begin{bmatrix} 2z(z+1) & 2z(z+1) \\ \lambda _++z^2(1+3z) & \lambda _-+z^2(1+3z)\end{bmatrix}\begin{bmatrix} -2z(z+1)\lambda _+^n \\ 2z(z+1)\lambda _-^n\end{bmatrix} \\
	&= \frac{1}{d}\begin{bmatrix} 4z^2(z+1)^2(\lambda _-^n-\lambda _+^n) \\ 2z(z+1)(\lambda _-^{n+1}-\lambda _+^{n+1})+2z^3(1+3z)(z+1)(\lambda _-^n-\lambda _+^n)\end{bmatrix}
\end{align*}
Here, $d$ is the determinant of this eigenvector matrix and $\lambda _\pm(z)$ are the eigenvalues of $M$. We may compute $\lambda _\pm$ to be as they are above. If we let $R_n(z)=\lambda _+^n(z)-\lambda _-^n(z)$ and take the quotient of the two entries of $M^nv_0$, our closed form expression for $r_n(z)$ results.
\end{proof}

It is worthwhile to note that though the closed form expression of $r_n$ involves square roots, it is still a rational function by its recursive form.

\begin{proposition}\label{prop:relation}
We have the following relation:
$$
\frac{1}{z}r_n(z)r_n\left(\frac{1}{z}\right) =
\frac{2}{3z^2-2z+3}\left[ r_n(z)+r_n\left(\frac{1}{z}\right)\right].
$$
\end{proposition}
\begin{proof}
First, let us compute $r_n\left(\frac{1}{z}\right)$. Note that $\lambda _\pm\left(\frac{1}{z}\right) =-\frac{1}{z^3}\lambda _\mp (z)$. From this it follows that $R_n\left(\frac{1}{z}\right) =-\left( -\frac{1}{z^3}\right) ^nR_n(z)$. Using this substitution we find $r_n\left(\frac{1}{z}\right)=\frac{2z(z+1)R_n(z)}{(z+3)R_n(z)-R_{n+1}(z)}$. Now consider the expression from the proposition:
$$
I_n(z):=\frac{1}{z}r_n(z)r_n\left(\frac{1}{z}\right)=
\frac{4z(z+1)^2R_n^2}{(R_{n+1}+z^2(1+3z)R_n)((z+3)R_n-R_{n+1})}.
$$
Let us assume the following partial fractions expansion:
$$
I_n(z)=\frac{a(z)}{R_{n+1}+z^2(1+3z)R_n}+\frac{b(z)}{((z+3)R_n-R_{n+1})}.
$$
This allows us to write:
$$
\left[ b(z)-a(z)\right] R_{n+1}(z)+\left[ (z+3)a(z)+z^2(1+3z)b(z)\right] R_n(z)
= 4z(z+1)^2R_n(z)^2.
$$
This equation holds if the following system has satisfied:
$$  b(z)-a(z)=0, \mbox{and} $$
$$  \left[ (z+3)a(z)+z^2(1+3z)b(z)\right]R_n(z)=4z(z+1)^2R_n(z)^2. $$
These equations imply
$
a(z)=b(z)=\frac{R_n(z)}{3z^2-2z+3}.
$
Plugging this back into $I_n(z)$ gives us the desired relation.
\end{proof}
\begin{proposition}\label{prop:imaginary}
For all $n\in\N$, $r_n(\omega )$ is purely imaginary
where $\omega =\frac{1}{3}+\frac{2\sqrt{2}}{3}i$.
\end{proposition}
\begin{proof}
This is immediately apparent by plugging $\omega$ into the recursion relation:
$
r_{n+1}(\omega ) =
\frac{2\sqrt{2}i-3r_n(\omega )}{3-2\sqrt{2}ir_n(\omega )}.
$
\end{proof}

We are now ready to state the proof of Theorem \ref{thm:recurrence}.
\begin{proof}
From \cite{titchmarsh1979theory},
we note the following representation governing the Hadamard product of two functions $f$ and $g$ (with radius of convergence $R$ and $R'$ respectively):
$$
(f\odot g)(z)=\frac{1}{2\pi i}\int _C\frac{1}{w}f(w)g\left(\frac{z}{w}\right) dw.
$$
Here, $C$ is a contour for which $|w|<R$ and $|\frac{z}{w}|<R'$. Clearly we have $p_n=(r_n\odot r_n)(1)$, so we may write:
$$
p_n=\frac{1}{2\pi i}\int _{|z|=1}\frac{1}{z}r_n(z)r_n\left(\frac{1}{z}\right) dz.
$$
By Corollary \ref{cor:analytical-function}, if $z_0$ is a pole of $r_n(z)$,
then it follows that $|z_0|>1$. Similarly, if $z_0$ is a pole of $r_n\left(\frac{1}{z}\right)$, then $|z_0|<1$. Thus, for every $n\in\N$ there exists an $\epsilon >0$ such that:
$$
p_n=\frac{1}{2\pi i}\int _{|z|=1+\epsilon}\frac{1}{z}r_n(z)r_n\left(\frac{1}{z}\right) dz.
$$
We substitute our relation from Proposition \ref{prop:relation}:
$$
p_n=\frac{1}{\pi i}\int _{|z|=1+\epsilon}\frac{r_n(z)}{3z^2-2z+3}dz+\frac{1}{\pi i}\int _{|z|=1+\epsilon}\frac{r_n\left(\frac{1}{z}\right)}{3z^2-2z+3}dz.
$$
Let $\omega =\frac{1}{3}+\frac{2\sqrt{2}}{3}i$ and $\bar{\omega}$ be roots of the equation $3z^2-2z+3=0$, and note that $|\omega |=|\bar{\omega}|=1$. The right integrand has all poles within the contour, and is a rational function with representation $\frac{p(z)}{q(z)}$ where $\text{deg}(q)=\text{deg}(p)+3$ by Proposition \ref{prop:rewrite-rz}.
A result from complex analysis tells us that for a rational function $\frac{p'(z)}{q'(z)}$ with $\text{deg}(q')>\text{deg}(p')+1$, the sum of the residues vanishes. As such, we are left with
$$
p_n=\frac{1}{\pi i}\int _{|z|=1+\epsilon}\frac{r_n(z)}{3z^2-2z+3}dz.
$$
The only poles enclosed by the contour are $\omega$ and $\bar{\omega}$. By Cauchy's integral formula, we thus have:
$$
p_n=\frac{i}{2\sqrt{2}}\left( r_n(\bar{\omega})-r_n(\omega )\right).
$$
By Proposition \ref{prop:imaginary}, $r_n(\omega )$ is purely imaginary.
Moreover, since $r_n$ is a rational function with real coefficients, we have $r_n(\bar{z})=\overline{r_n(z)}$. This implies the following simplification:
$
p_n=-\frac{i}{\sqrt{2}}r_n(\omega).
$
We can now readily prove the recurrence:
\begin{align*}
p_{n+1} &= -\frac{i}{\sqrt{2}}r_{n+1}(\omega ) \\
	&= -\frac{i}{\sqrt{2}}\left(\frac{2\omega (\omega +1)-\omega ^2(1+3\omega )r_n(\omega )}{(\omega +3)-2\omega (\omega +1)r_n(\omega )}\right) \\
	&= -\frac{i}{\sqrt{2}}\left(\frac{\left(\frac{-8+20\sqrt{2}i}{9}\right) +\left(\frac{10+2\sqrt{2}i}{3}\right) r_n(\omega )}{\left(\frac{10+2\sqrt{2}i}{3}\right) +\left(\frac{8-20\sqrt{2}i}{9}\right) r_n(\omega )}\right) \\
	&= -\frac{i}{\sqrt{2}}\left(\frac{\frac{2\sqrt{2}}{3}i+r_n(\omega )}{1-\frac{2\sqrt{2}}{3}ir_n(\omega )}\right) \\
	&= -\frac{i}{\sqrt{2}}\left(\frac{\frac{2\sqrt{2}}{3}i+i\sqrt{2}p_n}{1+\frac{4}{3}p_n}\right) \\
	&= \frac{2+3p_n}{3+4p_n}.
\end{align*}

\end{proof}

\end{appendices} 


\begin{thebibliography}{10}

\bibitem{aharonov1993quantum}
Yakir Aharonov, Luiz Davidovich, and Nicim Zagury.
\newblock Quantum random walks.
\newblock {\em Physical Review A}, 48(2):1687, 1993.

\bibitem{spitzer2013principles}
Frank Spitzer.
\newblock {\em Principles of random walk}, volume~34.
\newblock Springer Science \& Business Media, 2013.

\bibitem{ambainis2001one}
Andris Ambainis, Eric Bach, Ashwin Nayak, Ashvin Vishwanath, and John Watrous.
\newblock One-dimensional quantum walks.
\newblock In {\em Proceedings of the thirty-third annual ACM symposium on
  Theory of computing}, pages 37--49. ACM, 2001.

\bibitem{aharonov2001quantum}
Dorit Aharonov, Andris Ambainis, Julia Kempe, and Umesh Vazirani.
\newblock Quantum walks on graphs.
\newblock In {\em Proceedings of the thirty-third annual ACM symposium on
  Theory of computing}, pages 50--59. ACM, 2001.

\bibitem{moore2002quantum}
Cristopher Moore and Alexander Russell.
\newblock Quantum walks on the hypercube.
\newblock In {\em Randomization and Approximation Techniques in Computer
  Science}, pages 164--178. Springer, 2002.

\bibitem{childs2004spatial}
Andrew~M Childs and Jeffrey Goldstone.
\newblock Spatial search by quantum walk.
\newblock {\em Physical Review A}, 70(2):022314, 2004.

\bibitem{krovi2006hitting}
Hari Krovi and Todd~A Brun.
\newblock Hitting time for quantum walks on the hypercube.
\newblock {\em Physical Review A}, 73(3):032341, 2006.

\bibitem{aaronson2004quantum}
Scott Aaronson and Yaoyun Shi.
\newblock Quantum lower bounds for the collision and the element distinctness
  problems.
\newblock {\em Journal of the ACM (JACM)}, 51(4):595--605, 2004.

\bibitem{ambainis2007quantum}
Andris Ambainis.
\newblock Quantum walk algorithm for element distinctness.
\newblock {\em SIAM Journal on Computing}, 37(1):210--239, 2007.

\bibitem{szegedy2004quantum}
Mario Szegedy.
\newblock Quantum speed-up of markov chain based algorithms.
\newblock In {\em Foundations of Computer Science, 2004. Proceedings. 45th
  Annual IEEE Symposium on}, pages 32--41. IEEE, 2004.

\bibitem{childs2003exponential}
Andrew~M Childs, Richard Cleve, Enrico Deotto, Edward Farhi, Sam Gutmann, and
  Daniel~A Spielman.
\newblock Exponential algorithmic speedup by a quantum walk.
\newblock In {\em Proceedings of the thirty-fifth annual ACM symposium on
  Theory of computing}, pages 59--68. ACM, 2003.

\bibitem{childs2009universal}
Andrew~M Childs.
\newblock Universal computation by quantum walk.
\newblock {\em Physical review letters}, 102(18):180501, 2009.

\bibitem{lovett2010universal}
Neil~B Lovett, Sally Cooper, Matthew Everitt, Matthew Trevers, and Viv Kendon.
\newblock Universal quantum computation using the discrete-time quantum walk.
\newblock {\em Physical Review A}, 81(4):042330, 2010.

\bibitem{kempe2003quantum}
Julia Kempe.
\newblock Quantum random walks: an introductory overview.
\newblock {\em Contemporary Physics}, 44(4):307--327, 2003.

\bibitem{venegas2012quantum}
Salvador~Elias Venegas-Andraca.
\newblock Quantum walks: a comprehensive review.
\newblock {\em Quantum Information Processing}, 11(5):1015--1106, 2012.

\bibitem{inui2005one}
Norio Inui, Norio Konno, and Etsuo Segawa.
\newblock One-dimensional three-state quantum walk.
\newblock {\em Physical Review E}, 72(5):056112, 2005.

\bibitem{inui2005localization}
Norio Inui and Norio Konno.
\newblock Localization of multi-state quantum walk in one dimension.
\newblock {\em Physica A: Statistical Mechanics and its Applications},
  353:133--144, 2005.

\bibitem{inui2004localization}
Norio Inui, Yoshinao Konishi, and Norio Konno.
\newblock Localization of two-dimensional quantum walks.
\newblock {\em Physical Review A}, 69(5):052323, 2004.

\bibitem{vstefavnak2012continuous}
Martin {\v{S}}tefa{\v{n}}{\'a}k, I~Bezd{\v{e}}kov{\'a}, and Igor Jex.
\newblock Continuous deformations of the grover walk preserving localization.
\newblock {\em The European Physical Journal D}, 66(5):1--7, 2012.

\bibitem{vstefavnak2014stability}
Martin {\v{S}}tefa{\v{n}}{\'a}k, Iva Bezd{\v{e}}kov{\'a}, Igor Jex, and
  Stephen~M Barnett.
\newblock Stability of point spectrum for three-state quantum walks on a line.
\newblock {\em Quantum Information \& Computation}, 14(13-14):1213--1226, 2014.

\bibitem{vstefavnak2014limit}
M~{\v{S}}tefa{\v{n}}{\'a}k, I~Bezd{\v{e}}kov{\'a}, and Igor Jex.
\newblock Limit distributions of three-state quantum walks: the role of coin
  eigenstates.
\newblock {\em Physical Review A}, 90(1):012342, 2014.

\bibitem{motwani2010randomized}
Rajeev Motwani and Prabhakar Raghavan.
\newblock {\em Randomized algorithms}.
\newblock Chapman \& Hall/CRC, 2010.

\bibitem{grover1997quantum}
Lov~K Grover.
\newblock Quantum mechanics helps in searching for a needle in a haystack.
\newblock {\em Physical review letters}, 79(2):325, 1997.

\bibitem{bach2004one}
Eric Bach, Susan Coppersmith, Marcel~Paz Goldschen, Robert Joynt, and John
  Watrous.
\newblock One-dimensional quantum walks with absorbing boundaries.
\newblock {\em Journal of Computer and System Sciences}, 69(4):562--592, 2004.

\bibitem{falkner2014weak}
Stefan Falkner and Stefan Boettcher.
\newblock Weak limit of the three-state quantum walk on the line.
\newblock {\em Physical Review A}, 90(1):012307, 2014.

\bibitem{bach2009absorption}
Eric Bach and Lev Borisov.
\newblock Absorption probabilities for the two-barrier quantum walk.
\newblock {\em arXiv preprint arXiv:0901.4349}, 2009.

\bibitem{titchmarsh1979theory}
E.~C. Titchmarsh.
\newblock {\em The Theory of Functions: 2d Ed}.
\newblock Oxford University Press, 1979.

\end{thebibliography}
\end{document}